\documentclass[11pt]{article}
\usepackage{inputenc}
\usepackage{authblk}
\usepackage{tikz}
\usepackage{amsmath}
\usepackage{amssymb}
\usepackage[hidelinks]{hyperref}
\usepackage{algorithmic}
\usepackage{algorithm}

\newcounter{remark}[section]

\def\claim{\par\medskip\noindent\refstepcounter{remark}\hbox{\bf Remark \arabic{section}.\arabic{remark}}
	\ 
}
\def\endclaim{
	\par\medskip}
\newenvironment{remark}{\claim}{\endclaim}

\newcommand{\s}{\mathbf{s}}

\newtheorem{lemma}{Lemma}

\newtheorem{assumption}{Assumption}[section]
\newtheorem{thm}{Theorem}[section]

\def\endpf{{\ \hfill\hbox{\vrule width1.0ex height1.0ex}\parfillskip 0pt
	}}
	\newenvironment{proof}{\noindent{\bf Proof:}}{\endpf}
\begin{document}
		\title{Are Thousands of Samples Really Needed to Generate Robust Gene-List for Prediction of Cancer Outcome?}
		\author{Royi Jacobovic}\footnote{Comments may be sent to \textit{royi.jacobovic@mail.huji.ac.il.}}
		\maketitle
		
		\begin{abstract}
			The prediction of cancer prognosis and metastatic potential immediately after the initial diagnoses is a major challenge in current clinical research. The relevance of such a signature is clear, as it will free many patients from the agony and toxic side-effects associated with the adjuvant chemotherapy automatically and sometimes carelessly subscribed to them. Motivated by this issue, Ein-Dor (2006) and Zuk (2007) presented a Bayesian model which leads to the following conclusion: Thousands of samples are needed to generate a robust gene list for predicting outcome. This conclusion is based on existence of some statistical assumptions. The current work raises doubts over this determination by showing that: (1) These assumptions are not consistent with additional assumptions such as sparsity and Gaussianity. (2) The empirical Bayes methodology which was suggested in order to test the relevant assumptions doesn't detect severe violations of the model assumptions and consequently an overestimation of the required sample size might be incurred.   	
		\end{abstract}
	
		\section{Introduction}
		In clinical research, the need for sensitive and reliable predictors of outcome
		is most acute for early discovery of metastases. In the recent decades,
		gene-expression data is available and can be used for this purpose. From
		a statistical point of view, this kind of data is hard to analyze because the
		number of genes is up to tens of thousands while the measurements of the
		gene-expressions are associated with non-negligible observational noise. In
		particular, it may be hard to pinpoint the most predictive genes. Motivated
		by this issue, \cite{ein-dor} and \cite{zuk} suggested a Bayesian modelling of the observational
		noise and then, with regard to this modelling, the conclusion was that thousands
		of samples are needed to generate a robust gene-list for predicting outcome. The current work reveals that the statistical framework which was used in order to derive this conclusion is inconsistent with assumptions like sparsity and Gaussianity. To motivate this theoretical result, observe that sparsity and Gaussianity are  commonly assumed by practitioners. For instance, some applications of sparse models to micro-array data analysis are e.g. \cite{cara} and \cite{knowels}. Similarly, examples of applications which are using Gaussian models are given by e.g. \cite{ghosh} and \cite{yeung}.
		To complete the picture, another
		topic to be discussed is the empirical Bayes (EB) methodology which was implemented by \cite{ein-dor} and \cite{zuk} in order to test the relevant model assumptions. In this context, the current work presents the results of a simulation analysis which demonstrates a case such that
		the EB testing methodology doesn't detect severe violations of the model assumptions and consequently an
		overestimation of the needed sample size is incurred.
		The rest is organized as follows: Section \eqref{sec:2} presents a detailed
		description of the statistical model which was phrased by \cite{ein-dor} and \cite{zuk}. Section \eqref{sec:3}
		includes rigorous statements of the claims of this work along with their proofs. Namely, these claims specify the exact notions of sparsity and Gaussianity under
		which inconsistencies with the model assumptions take place. Section \eqref{sec:4} uses Monte-Carlo (MC) simulation in order to analyse a specific setup which is associated with severe violations of the model assumptions. Then, despite of these violations, it is shown that the above-mentioned EB approach for testing the model assumptions
		doesn't alert the user and as a result
		too pessimistic estimates of the needed sample-size are obtained.
		Finally, Section \eqref{sec:5} is a brief summary about the contributions of
		this work with suggestions for further research.
		
		\section{Model Description}\label{sec:2}
		Let $(X_{1j},\ldots,X_{kj}, Y_j )\stackrel{i.i.d}{\sim}
		(X_1, \ldots , X_k, Y ), j = 1, . . . , n$ be an i.i.d sample
		of $n=4,5,\ldots$ observations from some $(k + 1)$-dimensional parametric multivariate
		distribution $F_\theta; \theta\in\Theta \subseteq\mathbb{R}^m;m\in\mathbb{N}$, such that $\theta\sim G$ where $G$ is the prior distribution of the model. It is assumed that almost surely $F_\theta$ is associated with finite first two moments. Now, the statistical terminology that $X_1,\ldots , X_k$
		are features and $Y$ is the target variable is adapted. Notice that in the context of statistical analysis of micro-array datasets, the features are the genes and the target variable is the survival status. In
		particular, the datasets which were used by \cite{ein-dor} and \cite{zuk} are characterized by survival status which is
		a binary variable. With respect to this terminology, as pointed by \cite{donoho}, if the number of
		observations is less than the number of features, then the statistician confronts with the curses and blessings of dimensionality. One possible approach to handle such circumstances is to choose the best features to
		explain Y, i.e. to perform some feature selection procedure. In practice, such selection may be done according to the absolute values of the following Pearson coefficients of
		correlation
		\begin{equation*}
		\rho_i(\theta):=\rho(X_i,Y;\theta)=\frac{C_\theta(X_i,Y)}{\sqrt{V_\theta(X_i)V_\theta(Y)}} \ \ , \ \ i=1,\ldots,k
		\end{equation*}
		where $C_\theta(\cdot)$ and $V_\theta(\cdot)$ respectively symbol the covariance and variance operators with respect to the parametrization $\theta\in\Theta$. Few examples for applications of this class of procedures are \cite{guyon}, \cite{hall} and \cite{yu}. In-addition, in order to streamline the presentation of the contents of this work, the notation of $\theta$ beneath the variance and covariance operators is discarded. 
		
		Now, considering the fact that the statistician has no direct access to the true correlations, \cite{ein-dor} and \cite{zuk} considered a setup such that $k>>1$ ($k=20000$ in humans) and  suggested an easy to implement method to evaluate the number of observations which is required in order to obtain a robust list of features whose absolute values of the true correlations with the target variable are the highest. In details, they defined robust gene list as a list having at least 50\% overlap with the list of the $l$ genes with the highest absolute correlations with the survival status where $l$ is of the order of few dozens. To proceed, denote the sample analogues of
		$\rho_1,\ldots,\rho_k$ by $r_1^n,\ldots,r_k^n$
		and recall that Fisher's transformation (see \cite{fisher1915,fisher1921}) is defined by
		\begin{equation}
		\phi(h):=\frac{1}{2}\ln\frac{1+h}{1-h} \ \ , \ \ \forall h\in(-1,1)\,.
		\end{equation}
		The model assumptions are as follows:
		
		\newpage
		
		\begin{assumption} \label{ass:1}
			$G$ is such that $\phi(\rho_1),\ldots, \phi(\rho_k)\stackrel{i.i.d}{\sim}
			N(\theta, \sigma^2_
			q)$,
			where $\sigma^2_q\in(0,\infty)$ is known parameter. \footnote{
				\cite{ein-dor} assumes that the distribution of the Fisher transformations of the true correlations
				can be approximated by centred Gaussian distribution with variance $\sigma^2_q\in(0,\infty)$. Since the
				notion of approximation is not mentioned by \cite{ein-dor}, the above-mentioned Assumption \eqref{ass:1} is taken
				from \cite{zuk}. In addition, notice that \cite{zuk} considers a more general settings by letting $\phi(\rho_1), \ldots, \phi(\rho_k)
				\stackrel{i.i.d}{\sim}
				q(\cdot)$ where $q(\cdot)$ is a general density.
			}
		\end{assumption}
		
		\begin{assumption} \label{ass:2}
			For almost any parametrization $\theta\in\Theta$ w.r.t. $G$, $\phi(r_1^n),\ldots,\phi(r_k^n)$ are asymptotically
			independent random variables as $n\rightarrow\infty$.\footnote{Assumption \eqref{ass:2} is a relaxed version of the independence assumption which was made by \cite{ein-dor} and \cite{zuk}}
		\end{assumption}
		
		\begin{assumption} \label{ass:3}
			For almost any parametrization $\theta\in\Theta$ w.r.t. $G$, 
			$\sqrt{n-3}\big(\phi(r^n_i)-\phi(\rho_i)\big)\xrightarrow{\mathcal{L}}N(0,1),\forall i = 1, \ldots , k$.\footnote{The notation $\xrightarrow{\mathcal{L}}$ refers to convergence in law of a sequence of random variables. Exact definition is provided by \cite{ferg}.} \footnote{ Assumption \eqref{ass:3} is a relaxed version of the convergence assumption which was made by \cite{ein-dor} and \cite{zuk}.}
		\end{assumption}
		
		\section{Implications of Model Assumptions}\label{sec:3}
		To start with, denote the correlation between $X_i$ and
		$X_j$ $(i, j = 1, . . . , k)$ by $\rho_{x_ix_j}$. Using this notation, since features which are totally correlated are almost surely identical up to multiplication of a non-zero constant, there is no loss of generality by assuming that almost surely for all $1\leq i<j\leq k$, $|\rho_{x_ix_j}|<1$. Now, the next theorem states that if Assumption \eqref{ass:1} holds, then there is no pair of genes whose correlation is fixed almost surely. Thus, as a result of Assumption \eqref{ass:1}, any pair of genes is correlated with positive probability. 
		\begin{thm} \label{thm:sparsity}
			Let $c\in(-1, 1)$ . If $G$ is a prior distribution which satisfies Assumption
			\eqref{ass:1}, then for any $1\leq i<j\leq k$
			
			\begin{equation*}
			\mathbb{P}\{\rho_{x_ix_j}=c\}<1
			\end{equation*} 	
		\end{thm}
		
		\begin{proof}
			Due to symmetry considerations, it is enough to prove that 
			
			\begin{equation*}
			\mathbb{P}\{\rho_{x_1x_2}=c\}<1\ \ , \ \ \forall c\in(-1,1)\,. 
			\end{equation*}
			To this end, assume by contradiction that there exists some $c\in(-1, 1)$
			such that $\mathbb{P}\{\rho_{x_1x_2}=c\}=1$. Since $G$ is a probability distribution over $\Theta$ and it is known that almost surely $F_\theta$ is associated with finite first two moments,
			then the probability (with respect to $G$) that the correlation
			matrix of $(X_1,X_2,Y)$ is positive semi-definite, equals to one. To obtain a contradiction, it is shown that the characteristic polynomial of this correlation matrix is associated with negative roots with positive probability. In details, since
			$\mathbb{P}\{\rho_{x_1x_2}=c\}=1$, then almost surely the characteristic polynomial of the correlation matrix of $(X_1,X_2,Y)$ is given by
			\begin{equation*}
			P(\lambda;\rho_1,\rho_2,c)=\det\begin{bmatrix}
			1-\lambda & c & \rho_1 \\ 
			c & 1-\lambda & \rho_2 \\
			\rho_1 & \rho_2 & 1-\lambda
			\end{bmatrix}=
			\end{equation*}
			
			\begin{equation*}
			=(1-\lambda)^3-(1-\lambda)(\rho_1^2+\rho_2^2+c^2)-2c\rho_1\rho_2\,.
			\end{equation*}
			Now, set $\rho_1=0$ and obtain the following equation:
			\begin{equation*}
			P(\lambda;\rho_1=0,\rho_2,c)=(1-\lambda)^3-(1-\lambda)(c^2+\rho_2^2)=0 \,.
			\end{equation*}
			Clearly, if $\rho_2=
			\sqrt{\frac{2-c^2}{2}}\in(-1,1)$, then one root of $P(\lambda)$ is given by
			\begin{equation*}
			\hat{\lambda}=1-\sqrt{1+\frac{c^2}{2}}<0
			\end{equation*} 
			, i.e. there exists a negative solution to the equation 
			\begin{equation*}
			P\bigg(\lambda;\rho_1=0,\rho_2=\sqrt{\frac{2-c^2}{2}},c\bigg)=0\,.
			\end{equation*}
			Since Cardano formula \footnote{For details about Cardano formula, look at \textit{mathworld.wolfram.com/CubicFormula.html}.}  implies that
			the solutions of the equation $P(\lambda; \rho_1, \rho_2,c) = 0$ are continuous in $(\rho_1, \rho_2)$
			at the point $p:=\big(0,\sqrt{\frac{2-c^2}{2}}\big)$, there exists $\delta>0$ which is associated with a ball $B_\delta(p)\subseteq(-1,1)^2$ such that for any $(\rho_1,\rho_2)\in B_\delta(p)$ there exists a negative root for $P(\lambda;\rho_1,\rho_2,c)$. In addition, the fact that $\phi(\cdot)$ is strictly increasing continuous function and $\phi(\rho_1), \phi(\rho_2)\stackrel{i.i.d}{\sim}N(0, \sigma^2_q)$ implies that $(\rho_1,\rho_2)$ are continuously distributed over $(-1, 1)^2$. Therefore, with positive probability, $P(\lambda;\rho_1,\rho_2,c)$ is associated with negative root.   
			
		\end{proof}
		
		\begin{lemma} (\textbf{Multivariate Delta Method})
			If $\Psi$ is $m$-dimensional positive definite matrix
			and $\{\hat{\beta}_n\}_{n=1}^\infty$ is a consistent sequence of estimators of a parameter vector $\beta\in\mathbb{R}^m$ such that
			\begin{equation*}
			\sqrt{n}\big(\hat{\beta}_n-\beta\big)\xrightarrow{\mathcal{L}}N(0,\Psi)
			\end{equation*}
			, then for any differentiable function $f:\mathbb{R}^m\rightarrow\mathbb{R}^d$, the following convergence holds:
			\begin{equation*}
			\sqrt{n}\big(f(\hat{\beta}_n)-f(\beta)\big)\xrightarrow{\mathcal{L}}N\bigg(0,[\nabla f(\beta)]^T\Psi[\nabla f(\beta)]\bigg)
			\end{equation*}
			where $\nabla f(\tilde{\beta})$ is the partial derivative matrix of $f(\cdot)$ at the point $\tilde{\beta}\in\mathbb{R}^m$.
		\end{lemma}
		\begin{proof}
			See chapter 7 of \cite{ferg}.
		\end{proof}
		\begin{thm} \label{thm:sufficient condition}
			If 
			$(X_1,\ldots,X_k,Y;\theta)\sim F_\theta$ such that almost surely $F_\theta$ is associated with finite first
			four moments and mean zero, then Assumption \eqref{ass:2} holds iff the event that for any $1 \leq i <
			j \leq k$ 
			
			\begin{equation} \label{eq:sufficient conditions}
			-\frac{\rho_j}{2\sigma^2_{x_j}}\bigg[-\frac{\rho_i}{2\sigma^2_{x_i}}C(X_i^2,X_j^2)-\frac{\rho_i}{2\sigma^2_y}C(Y^2,X_j^2)+\frac{1}{\sigma_{x_i}\sigma_y
			}C(X_iY,X_j^2)\bigg]-
			\end{equation}
			\begin{equation*}
			-\frac{\rho_j}{2\sigma^2_y}\bigg[-\frac{\rho_i}{2\sigma^2_{x_i}}C(X_i^2,Y^2)-\frac{\rho_i}{2\sigma^2_y}C(Y^2,Y^2)+\frac{1}{\sigma_{x_i}\sigma_y
			}C(X_iY,Y^2)\bigg]+
			\end{equation*}
			\begin{equation*}
			+\frac{1}{\sigma_{x_j}\sigma_y}\bigg[-\frac{\rho_i}{2\sigma^2_{x_i}}C(X_i^2,X_jY)-\frac{\rho_i}{2\sigma^2_y}C(Y,X_jY)+\frac{1}{\sigma^2_{x_i}\sigma_y
			}C(X_iY,X_jY)\bigg]=0
			\end{equation*}
			
			occurs with probability one.
			
		\end{thm} 
		\begin{proof}
			Consider the following notations:
			\begin{equation*}
			m_{x_i}:=\frac{1}{n}\sum_{j=1}^{n}X_{ij} \ \ , \ \ \forall i=1,\ldots,k
			\end{equation*}
			\begin{equation*}
			m_{x_i^2}:=\frac{1}{n}\sum_{j=1}^{n}X^2_{ij} \ \ , \ \ \forall i=1,\ldots,k
			\end{equation*}
			\begin{equation*}
			m_{x_iy}:=\frac{1}{n}\sum_{j=1}^{n}X_{ij}Y_j \ \ , \ \ \forall i=1,\ldots,k
			\end{equation*}
			\begin{equation*}
			m_{y}:=\frac{1}{n}\sum_{j=1}^{n}Y_j
			\end{equation*}
			\begin{equation*}
			m_{y^2}:=\frac{1}{n}\sum_{j=1}^{n}Y_j^2
			\end{equation*} 
			\begin{equation*}
			s^2_{x_i}:=m_{x_i^2}-m^2_{x_i} \ \ , \ \ \forall i=1,\ldots,k
			\end{equation*} 
			\begin{equation*}
			s_{x_iy}:=m_{x_iy}-m_{x_i}m_{y}\ \ , \ \ \forall i=1,\ldots,k
			\end{equation*} 
			\begin{equation*}
			s^2_y:=m_{y^2}-m^2_{y}\,.
			\end{equation*} 
			Using these notations, the sample correlations can be written as:
			\begin{equation*}
			r_i:=r_n^i=\frac{s_{x_iy}}{s_{x_i}s_y}\ \ ,i=1,\ldots, k\,.
			\end{equation*}
			It is given that almost surely, $F_\theta$ is associated with finite first four moments and hence the multivariate central limit theorem implies that 
			\begin{equation*}
			\sqrt{n}\begin{bmatrix}
			\begin{pmatrix}m_{x_1}\\\vdots\\m_{x_k}\\m_y\\m_{x_1^2}\\\vdots\\m_{x_k^2}\\m_{y^2}\\m_{x_1y}\\\vdots\\m_{x_ky}
			\end{pmatrix}-\begin{pmatrix}
			0\\\vdots\\0\\0\\\sigma_{x_1}^2\\\vdots\\\sigma_{x_k}^2\\\sigma_y^2\\\sigma_{x_1y}\\\vdots\\\sigma_{x_ky}
			\end{pmatrix}
			\end{bmatrix}\xrightarrow[n\rightarrow\infty]{\mathcal{L}}N_{3k+2}(0,\Sigma^1) \ \ , \ \ \mathbb{P}-a.s.
			\end{equation*}
			where the covariance matrix $\Sigma^1$ is given by $\Sigma_{ij}^1=C(Z_i,Z_j), 1\leq i,j\leq 3k+2$ and $Z$ is as follows: 
			\begin{equation*}
			Z=(X_1,\ldots,X_k,Y,X_1^2,\ldots,X_k^2,Y^2,X_1Y,\ldots,X_kY)^T\,.
			\end{equation*}
			Define a function $\eta:\mathbb{R}^{3k+2}\rightarrow\mathbb{R}^{2k+1}$ by
			\begin{equation*}
			\eta(z)=\begin{pmatrix}
			z_{k+1+1}-z^2_{1}\\\vdots\\z_{k+1+k+1}-z^2_{k+1}\\z_{2k+2+1}-z_{k+1}z_1\\\vdots\\ z_{2k+2+k}-z_{k+1}z_k
			\end{pmatrix}
			\end{equation*}
			and notice that $\eta(\cdot)$ satisfies 
			\begin{enumerate}
				\item 
				\begin{equation*}
				\eta\begin{bmatrix}
				\begin{pmatrix}m_{x_1}\\\vdots\\m_{x_k}\\m_y\\m_{x_1^2}\\\vdots\\m_{x_k^2}\\m_{y^2}\\m_{x_1y}\\\vdots\\m_{x_ky}
				\end{pmatrix}
				\end{bmatrix}=\begin{pmatrix}
				s^2_{x_1}\\\vdots\\s^2_{x_k}\\s^2_y\\s_{x_1y}\\\vdots\\s^2_{x_ky}
				\end{pmatrix}
				\end{equation*}
				\item
				\begin{equation*}
				\nabla \eta(z)=\begin{pmatrix}
				-2diag(z_1,\ldots,z_{k+1}) & B  \\ I_{k+1\times k+1} & O_{k+1\times k} \\
				0_{k\times k+1} & I_{k\times k}
				\end{pmatrix}
				\end{equation*}
				where the matrix $B$ is given by
				\begin{equation*}
				B:=\begin{pmatrix}
				-z_{k+1}\cdot I_{k\times k} \\
				-z_1,\ldots,-z_k
				\end{pmatrix}\,.
				\end{equation*} 
			\end{enumerate} 
			At this stage, apply the multivariate delta method in order to obtain the limit
			\begin{equation*}
			\sqrt{n}\begin{bmatrix}
			\begin{pmatrix}
			s^2_{x_1}\\\vdots\\s^2_{x_k}\\s^2_y\\s_{x_1y}\\\vdots\\s^2_{x_ky}
			\end{pmatrix}-\begin{pmatrix}
			\sigma^2_{x_1}\\\vdots\\\sigma^2_{x_k}\\\sigma^2_y\\\sigma_{x_1y}\\\vdots\\\sigma^2_{x_ky}
			\end{pmatrix}
			\end{bmatrix}\xrightarrow[n\rightarrow\infty]{\mathcal{L}}N_{2k+1}(0,\Sigma^2) \ \ , \ \ \mathbb{P}-a.s.
			\end{equation*}
			where $\Sigma^2$ is given by
			\begin{equation*}
			\Sigma^2=\nabla \eta^T\begin{bmatrix}
			\begin{pmatrix}
			0\\\vdots\\0\\0\\\sigma_{x_1}^2\\\vdots\\\sigma_{x_k}^2\\\sigma_y^2\\\sigma_{x_1y}\\\vdots\\\sigma_{x_ky}
			\end{pmatrix}
			\end{bmatrix}\Sigma^1\nabla \eta\begin{bmatrix}
			\begin{pmatrix}
			0\\\vdots\\0\\0\\\sigma_{x_1}^2\\\vdots\\\sigma_{x_k}^2\\\sigma_y^2\\\sigma_{x_1y}\\\vdots\\\sigma_{x_ky}
			\end{pmatrix}
			\end{bmatrix}\,.
			\end{equation*}
			Notice that for the  vector of inputs written above $\nabla\eta$ is given by 
			
			\begin{equation*}
			\begin{pmatrix}
			O_{k+1\times k+1} & O_{k+1\times k} \\
			I_{k+1\times k+1} & O_{k+1\times k} \\ O_{k\times k+1} & I_{k\times k}
			\end{pmatrix}
			\end{equation*}
			
			, i.e.  $\Sigma^2$ equals to the down-right $2k+1\times2k+1$ block of $\Sigma^1$. Considering this result, define another function $\gamma:\mathbb{R}^{2k+1}\rightarrow\mathbb{R}^k$ as follows
			
			\begin{equation*}
			\gamma(v)=\begin{pmatrix}
			\frac{v_{k+1+1}}{\sqrt{v_{k+1}v_1}}\\\vdots\\\frac{v_{k+1+k}}{\sqrt{v_{k+1}v_k}}
			\end{pmatrix}
			\end{equation*}
			
			and observe that 
			
			\begin{enumerate}
				\item
				\begin{equation*}
				\gamma\begin{bmatrix}
				\begin{pmatrix}
				s^2_{x_1}\\\vdots\\s^2_{x_k}\\s^2_y\\s_{x_1y}\\\vdots\\s^2_{x_ky}
				\end{pmatrix}
				\end{bmatrix}=\begin{pmatrix}
				r_1\\\vdots\\r_k
				\end{pmatrix}
				\end{equation*}
				\item
				\begin{equation*}
				\nabla\gamma^T(v)=\Big[A|B|C\Big]
				\end{equation*}
				where the matrices $A,B$ and $C$ are given by
				\begin{equation*}
				\ \ \ \ \ \   A:=-diag\Big(\frac{v_{k+1+1}}{2\sqrt{v_1^3v_{k+1}}},\ldots,\frac{v_{k+1+k}}{2\sqrt{v_k^3v_{k+1}}}\Big)
				\end{equation*}
				\begin{equation*}
				\ \ B:=-\Big(\frac{v_{k+1+1}}{2\sqrt{v_1v^3_{k+1}}},\ldots,\frac{v_{k+1+k}}{2\sqrt{v_kv^3_{k+1}}}\Big)^T
				\end{equation*}
				\begin{equation*}
				C:=diag\Big(\frac{1}{\sqrt{v_1v_{k+1}}},\ldots,\frac{1}{\sqrt{v_1v_{k+1}}})\,.
				\end{equation*}
			\end{enumerate}
			Therefore, by using the multivariate delta method once again, obtain the limit
			\begin{equation*}
			\sqrt{n}\begin{bmatrix}
			\begin{pmatrix}
			r_1\\\vdots\\r_k
			\end{pmatrix}-\begin{pmatrix}
			\rho_1\\\vdots\\\rho_k
			\end{pmatrix}
			\end{bmatrix}\xrightarrow[n\rightarrow\infty]{\mathcal{L}}N_k[0,\Sigma^3] \ \ , \ \ \mathbb{P}-a.s.
			\end{equation*}
			where $\Sigma^3$ is given by
			
			\begin{equation*}
			\Sigma^3=\nabla \gamma^T\begin{bmatrix}
			\begin{pmatrix}
			\sigma_{x_1}^2\\\vdots\\\sigma_{x_k}^2\\\sigma_y^2\\\sigma_{x_1y}\\\vdots\\\sigma_{x_ky}
			\end{pmatrix}
			\end{bmatrix}\Sigma^2\nabla \gamma\begin{bmatrix}
			\begin{pmatrix}
			\sigma_{x_1}^2\\\vdots\\\sigma_{x_k}^2\\\sigma_y^2\\\sigma_{x_1y}\\\vdots\\\sigma_{x_ky}
			\end{pmatrix}
			\end{bmatrix}\,.
			\end{equation*}
			Define $\phi^*:(-1,1)^k\rightarrow\mathbb{R}^k$ as follows
			\begin{equation*}
			\phi^*(w):=\big(\phi(w_1),\ldots,\phi(w_k)\big)^T
			\end{equation*}
			and notice that $\phi(\cdot)$ is differentiable in its domain and hence its derivative matrix is given by
			\begin{equation*}
			\nabla\phi^*(w)=diag[\phi'(w_1),\ldots,\phi'(w_k)]\,.
			\end{equation*}
			If so, one more execution of the multivariate delta method implies that
			\begin{equation*}
			\sqrt{n}\begin{bmatrix}
			\begin{pmatrix}
			\phi(r_1)\\\vdots\\\phi(r_k)
			\end{pmatrix}-\begin{pmatrix}
			\phi(\rho_1)\\\vdots\\\phi(\rho_k)
			\end{pmatrix}
			\end{bmatrix}\xrightarrow[n\rightarrow\infty]{\mathcal{L}}N_K[0,\Sigma^4] \ \ , \ \ \mathbb{P}-a.s.
			\end{equation*}
			where $\Sigma^4$ is given by
			\begin{equation*}
			[\Sigma^4]_{ij}=\phi'(\rho_i)\phi'(\rho_j)[\Sigma^3]_{ij}\ \ , \ \  1\leq i,j\leq k\,.
			\end{equation*}	
			Here, it can be seen that $\phi'(\cdot)$ is positive for any possible input and hence, because non-correlation is equivalent to independence under Gaussian law, then for any $i\neq j$ asymptotic independence of $\phi(r_i)$ and $\phi(r_j)$ is equivalent to $[\Sigma^3]_{ij}=0$. To see how the needed result stems from this understanding, for simplicity and w.l.o.g, consider the case where $i=1$ and $j=2$. In this case $r_1$ and $r_2$ are asymptotically independent iff the following equation holds
			\begin{equation*} 
			\Sigma^3_{12}=\begin{pmatrix}
			-\frac{\rho_1}{2\sigma^2_{x_1}}\\ 0 \\ \vdots \\ 0 \\-\frac{\rho_1}{2\sigma^2_y}\\  \frac{1}{\sigma_{x_2}\sigma_y} \\ 0 \\ \vdots \\ 0
			\end{pmatrix}^T
			\Sigma^2\begin{pmatrix}
			0\\ -\frac{\rho_2}{2\sigma^2_{x_2}} \\ \vdots \\ 0 \\-\frac{\rho_2}{2\sigma^2_y}\\  0 \\ \frac{1}{\sigma_{x_2}\sigma_y} \\ \vdots \\ 0
			\end{pmatrix}=0\,.
			\end{equation*}
			
		\end{proof}
		
		\begin{remark}
			In fact, Theorem \eqref{thm:sufficient condition} specifies sufficient conditions under which almost surely $\phi(r_n^1),\ldots,\phi(r_n^k)$ are asymptotically independent univariate Gaussians.
		\end{remark}
		\begin{thm} \label{thm:gaussianity}
			Let $\Sigma \sim G$ and $(X_1,\ldots, X_k,Y|\Sigma)\sim N_{k+1}(0,\Sigma)$. If $G$
			satisfies Assumption \eqref{ass:1}, then Assumption \eqref{ass:2} is violated with positive probability.	
		\end{thm}
		\begin{proof}
			For simplicity and w.l.o.g. it is enough to show that the event of having $r^1_n$ and $r^2_n$ which are not asymptotically independent occurs with positive probability. To do so, consider $i=1$ and $j=2$, and notice that due to the previous theorem, it is enough to prove that Equation \eqref{eq:sufficient conditions} doesn't hold with positive probability. Now, $(X_1,X_2,Y)$ is a Gaussian and hence, as was shown by  \cite{Isserlis1918}, each of the covariances appeared in Equation \eqref{eq:sufficient conditions} can be expressed as follows
			
			\begin{equation*}
			C(X_1^2,X_2^2)=\mathbb{E}(X_1^2X_2^2)-\mathbb{E}(X_1^2)\mathbb{E}(X_2^2)=\sigma_{x_1}^2\sigma_{x_2}^2+2\sigma_{x_1x_2}^2-\sigma_{x_1}^2\sigma_{x_2}^2=2\sigma_{x_1x_2}^2  \ \ \ \ \ \ \ \ \ \ \ \ \ \ \ \ \ \ \ \
			\end{equation*}
			\begin{equation*}	C(Y,X_2^2)=\ldots=2\sigma_{x_2y}^2 \ \ \ \ \ \ \ \ \ \ \ \ \ \ \ \ \ \ \ \ \ \ \ \ \ \ \ \ \ \ \ \ \ \ \ \ \ \ \ \ \ \ \ \ \ \ \ \ \ \ \ \ \ \ \ \ \ \ \ \  \ \ \ \ \  \ \ \  \ \ \ \ \ \ \ \ \ \ \ \ \ \ \ \ \ \ \ \ \ \ 
			\end{equation*}
			\begin{equation*}
			C(X_1Y,X_2^2)=\mathbb{E}(X_1YX_2^2)-\mathbb{E}(X_1Y)\mathbb{E}(X_2^2)=\sigma_{x_2}^2\sigma_{x_1y}+2\sigma_{x_1x_2}\sigma_{x_2y}-\sigma_{x_1y}\sigma_{x_2}^2=2\sigma_{x_1x_2}\sigma_{x_2y}
			\end{equation*}
			\begin{equation*}
			C(X_1^2,Y)=\ldots=2\sigma_{x_1y}^2 \ \ \ \ \ \ \ \ \ \ \ \ \ \ \ \ \ \ \ \ \ \ \ \ \ \ \ \ \ \ \ \ \ \ \ \ \ \ \ \ \ \ \ \ \ \ \ \ \ \ \ \ \ \ \ \ \ \ \ \  \ \ \ \ \  \ \ \ \ \ \ \ \ \ \ \ \ \ \ \ \ \ \ \ \ \ \ \ \ \ 
			\end{equation*}
			\begin{equation*}
			C(Y^2,Y^2)=\mathbb{E}Y^4-\mathbb{E}^2\mathbb{E}Y^2=3\sigma_y^4-\sigma_y^2\sigma_y^2=2\sigma_y^4 \ \ \ \ \ \ \ \ \ \ \ \ \ \ \ \ \ \ \ \ \ \ \ \ \ \ \ \ \ \ \ \ \ \ \ \ \ \ \ \ \ \ \ \ \ \ \ \ \ \ \ \ \ \ \ \ \ \ 
			\end{equation*}
			\begin{equation*}
			C(X_1Y,Y^2)=\mathbb{E}(X_1Y^3)-\mathbb{E}(X_1Y)\mathbb{E}(Y^2)=3\sigma_y^2\sigma_{x_1y}-\sigma_{x_1y}\sigma^2_y=2\sigma^2_y\sigma_{x_1y} \ \ \ \ \ \ \ \ \ \ \ \ \ \ \ \ \ \ \ \ \ \ \ \ \ \ \ \ \ \ \ \ \ \ \ \ \ \ \ \ \ \ \ \ \ \ \ \ \ \ \ \ \ \ \ \ \ \ \ \  \ \ \ \ \  \ \ \ \ \ \  
			\end{equation*}
			\begin{equation*}
			C(X_1^2,X_2Y)=\ldots=2\sigma_{x_1x_2}\sigma_{x_1y}^2 \ \ \ \ \ \ \ \ \ \ \ \ \ \ \ \ \ \ \ \ \ \ \ \ \ \ \ \ \ \ \ \ \ \ \ \ \ \ \ \ \ \ \ \ \ \ \ \ \ \ \ \ \ \ \ \ \ \ \ \  \ \ \ \ \  \ \ \ \ \ \ \ \ \ \ \ \ \ \ \ \ \ \ \ \ \ \ \ \
			\end{equation*}
			\begin{equation*}
			C(Y,X_2Y)=\ldots=2\sigma_y^2\sigma_{x_2y} \ \ \ \ \ \ \ \ \ \ \ \ \ \ \ \ \ \ \ \ \ \ \ \ \ \ \ \ \ \ \ \ \ \ \ \ \ \ \ \ \ \ \ \ \ \ \ \ \ \ \ \ \ \ \ \ \ \ \ \  \ \ \ \ \  \ \ \ \ \ \ \ \ \ \ \ \ \ \ \ \ \ \ \ \ \ \ \ \ \ \ \ \ \ \ 
			\end{equation*}
			\begin{equation*}
			C(X_1Y,X_2Y)=\mathbb{E}(X_1X_2Y^2)-\mathbb{E}(X_1Y)\mathbb{E}(X_2Y)=\sigma^2_y\sigma_{x_1x_2}+2\sigma_{x_1y}\sigma_{x_2y}-\sigma_{x_1y}\sigma_{x_2y}= \ \ \ \ \ \ \ \ \ \ \ \ \ \ \ \ 
			\end{equation*}
			\begin{equation*}
			=\sigma_y^2\sigma_{x_1x_2}+\sigma_{x_1y}\sigma_{x_2y}
			\end{equation*}
			where $\sigma_{x_1}^2:=V(X_1)$,  $\sigma_{x_2}^2:=V(X_2)$,   $\sigma_{x_1x_2}:=C(X_1,X_2)$,  $\sigma_{x_1y}:=C(X_1,Y)$ and $\sigma_{x_2y}:=C(X_2,Y)$. 
			By insertion of these expressions into Equation \eqref{eq:sufficient conditions}, a sufficient and necessary condition for asymptotic independence of $r_n^1$ and $r_n^2$ is given by:

			\begin{equation} \label{eq:sufficient 2}
			\rho_{x_1x_2}^2\frac{\rho_1\rho_2}{2}+\rho_{x_1x_2}(1-\rho_1^2-\rho_2^2)+\frac{\rho_1\rho_2^3+\rho_1^3\rho_2-\rho_1\rho_2}{2}=0
			\end{equation}
			
			where $\rho_{x_1x_2}:=\sigma_{x_1,x_2}/\sqrt{\sigma^2_{x_1}\sigma^2_{x_2}}$. The next step is to show that with positive probability, $(\rho_1,\rho_2)\in(-1,1)^2$ is such that Equation \eqref{eq:sufficient 2} has no solution. To see this, since $\phi(x)=0$ iff $x=0$, then Assumption \eqref{ass:1} implies that $\mathbb{P}\{\rho_i\neq0,\forall i=1,2\}=1$. Therefore, Equation \eqref{eq:sufficient 2} is almost surely a quadratic equation w.r.t. $\rho_{x_1x_2}$ that, depending on the values of $\rho_1$ and $\rho_2$, might not have a solution. Indeed, if $(\rho_1,\rho_2)=(0.5,0.9)\in(-1,1)^2$, then the discriminant of the quadratic equation equals to $-0.0085<0$.  
			
			Now, the fact that the discriminant of the quadratic equation is continuous in $\rho_1$ and $\rho_2$ at the point (0.5,0.9) implies that there exists some $\delta>0$ such that the discriminant is negative for any $(\rho_1,\rho_2)\in B_{\delta}(0.5,0.9)\subset(-1,1)^2$. By Assumption \eqref{ass:1}, $G$ is a prior such that $\phi(\rho_1),\ldots,\phi(\rho_k)\stackrel{iid}{\sim}N(0,\sigma_q^2)$ and hence
			\begin{equation*}
			\mathbb{P}\bigg(\{\phi(\rho_1),\phi(\rho_2)\}\subset
			\phi\big(B_\delta(0.5,0.9)\big)\bigg)>0 \ \ \footnotetext{For any set $\mathbf{X}\subseteq(-1,1)^2$, the notation $\phi(\mathbf{X})$ refers to the image of $\mathbf{X}$ by the Fisher transformation $\phi(\cdot)$.}
			\end{equation*}
			where $\mathbb{P}(\cdot)$ is the probability measure which is associated with the distribution $G(\cdot)$. Since $\phi(\cdot)$ is strictly increasing, there exists a strictly monotone inverse $\phi^{-1}(\cdot)$ which means that    
			\begin{equation*}
			\mathbb{P}\bigg(\{\rho_1,\rho_2\}\subset B_\delta(0.5,0.9)\bigg)>0
			\end{equation*}
			, i.e. with positive probability there is no solution for Equation \eqref{eq:sufficient 2}. 	
			
		\end{proof}
		
		\begin{remark}
			Theorems \eqref{thm:sufficient condition} and \eqref{thm:gaussianity}  were done regarding the special case where $\int_{\mathbb{R}^{k+1}} xF_\theta(dx)=\underline{0},\forall\theta\in\Theta$. In practice, the data is normalized and hence this assumption naturally stems from practical conventions.  
		\end{remark}
		
		\begin{remark}
			Theorems \eqref{thm:sparsity} and \eqref{thm:gaussianity} remain valid even when Assumption \eqref{ass:1} is phrased as follows: $\phi(\rho_1),\ldots , \phi(\rho_k)\stackrel{i.i.d}{\sim}q(\cdot)$
			such that $q(\cdot)$ is a density function which is supported on $\mathbb{R}$.
		\end{remark}
		
		\section{Testing Model Assumptions}\label{sec:4}
		Since the determination of \cite{ein-dor} and \cite{zuk} regarding the huge extent of the needed sample size is based on the existence of the Assumptions \eqref{ass:1}-\eqref{ass:3}, they implemented an EB methodology to validate these assumptions by empirical data. This testing methodology is described in the supplementary materials of \cite{ein-dor} as well as in Section 5 of \cite{zuk}. Practically, it suggests to conduct a visual checking to see whether the empirical distribution
		of $\phi(r^1_n),\ldots,\phi(r^k_n)$ looks like a Gaussian. With respect to this methodology, the current section
		presents an example of a model which strongly violates the setup of Section \eqref{sec:2} but on the same time generates Fisher-transformed sample correlations whose empirical
		distribution seems Gaussian. Then, it is shown that, for this specific model,
		application of the suggested methodology in order to evaluate the required
		sample-size returns too pessimistic evaluation.
		\subsection{Model Setup} \label{subsec:model-setup}
		Consider the case where there are $k>>1$ genes. In addition,
		let $\beta=(\beta_1,\ldots, \beta_k)$ be a vector which is distributed uniformly over the set of $k$ dimensional vectors that include $0<u<<k$ ones and $k-u$
		zeros ($u$ is known). Then, let $X=(X_1,\ldots,X_k|\beta)\sim N_k(0, I)$ and set
		$Y |(\beta,X)=\sum_{i=1}^k\beta_iX_i$. 
		
		Now, it is an immediate insight that this model strongly violates Assumption \eqref{ass:1}
		because the distribution of
		$\phi(\rho_1),\ldots,\phi(\rho_k)$
		is not continuous. However, as it
		seems from Figure (1), $\phi(r^1_n),\ldots,\phi(r^k_n)$ are distributed according to
		some centred Gaussian law and hence, the testing methodology states that the model assumptions may
		be carried out.
		\newpage
		\begin{remark}
			Notice that this model captures the following characteristics of gene-expression datasets:
			
			\begin{enumerate}
				\item There are many genes, i.e. $k>>1$ (20000 in humans)
				
				\item Only a small fraction of the genes are correlated with the target variable.
				
				\item  Those genes that are correlated with the target variable are associated
				with low values of absolute correlation.
			\end{enumerate} 
		\end{remark}
		
		\subsection{Straightforward vs. Approximated Computation}
		Generally speaking, given a dataset, i.e. a set of $n$ i.i.d realizations from the Bayesian model depicted by Subsection \eqref{sec:2}, a reasonable procedure to pinpoint the $u$ genes which are associated with the positive values of absolute correlation
		with the target variable is to compute the absolute values of the sample-correlations of all genes with
		the target variable and pick the $u$ genes whose absolute sample correlations are the highest. With regard to this selection procedure, the goal is to provide estimates of the expectation and standard deviation of the
		proportion of genes that are selected correctly. However, as pointed out by \cite{ein-dor} and \cite{zuk}, this calculation isn't trivial analytically. Moreover, even if the model setup is quite simple, straightforward MC simulation may require non-negligible running time. Therefore, since it has already been showed that the testing methodology wrongly stated that the simulated data from the model of subsection \eqref{subsec:model-setup} satisfies these assumptions, they suggested a fast approximated approach to calculate these estimates. Figure (2) exhibits a comparison between the results of a straightforward MC estimates and the fast approximated approach and indeed, it shows that an overestimation of the needed sample size is incurred.
		
		\section{Summary and Further Research}\label{sec:5}
		A very interesting feature of the model presented by Section \eqref{sec:2}, is the way in which it is defined by indirect assumptions over the data generating process (DGP). This work shows how the class of possible DGP's may be extracted from such an indirect setup. In addition, it has been demonstrated that the methodology of \cite{ein-dor} and \cite{zuk} to test the model assumptions may not detect severe violations of Assumptions \eqref{ass:1}-\eqref{ass:3}. These findings lead to the following directions for further research:
		
		\begin{enumerate}
			\item Development of Bayesian models which satisfy Assumptions \eqref{ass:1}-\eqref{ass:3} and are not too complex from a statistical point of view. Especially, besides the mathematical requirements which must be satisfied, they should also be flexible in the sense that they capture the essence behind the informative richness of gene-expression micro-array datasets.
			
			\item Development of better methodologies for testing Assumptions \eqref{ass:1}-\eqref{ass:3}. With respect to this point, notice that such development should be done in the context of the EB literature. Otherwise, if the classical Bayesian framework is adapted, then such prior assumptions are part of a belief system which conceptually can't be checked empirically.
			
			\item Comparative research between the methodology presented by \cite{ein-dor} and \cite{zuk} and the techniques
			which are studied by the literature of ranking and selection (R\&S) procedures. Generally speaking, this literature investigates procedures for ranking and selection from stochastic populations by their statistical properties
			such as mean, variance, R-squared with a target variable, etc. To see
			the relevancy of this literature, observe that the question of ranking
			and selection of features by their absolute correlations (or equivalently by their R-squared) with some target variable
			was investigated by \cite{alam}, \cite{alam76}, \cite{levy75}, \cite{levy77}, \cite{ramberg77}, \cite{rizvi73} and \cite{wilcox} in the context of this literature.
			To motivate such a comparative research, there exists a modern literature, e.g. \cite{cui08} and \cite{cui10}, that raised the question of how to apply R\&S theory to gene-expression micro-array datasets? 
			
			Finally, all the simulation results that were exhibited here are done by R program whose code is available at \textbf{\textit{https://github.com/royija/thousands-of-samples}.} 
		\end{enumerate}
		
		\bibliographystyle{abbrv}
		\bibliography{thousands}

		\begin{figure}[p!]\label{fig:hist}
			\begin{center}
				\includegraphics[height=3.0in, width=5in]{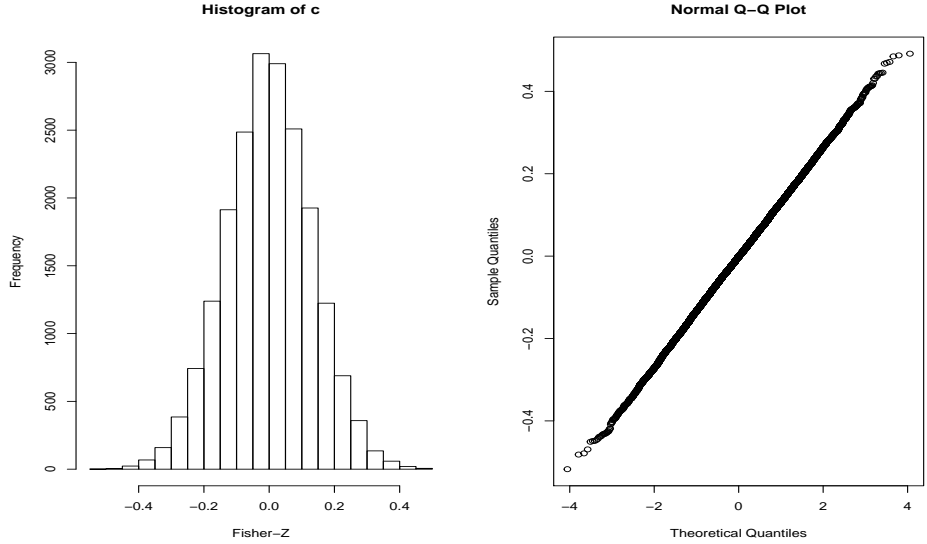} \caption{Histogram and normal QQ-plot of $n=59$ i.i.d observations of the model described by Subsection \eqref{subsec:model-setup} where $u=100$ and $k=20000$. 
				}
			\end{center}
		\end{figure}

		\begin{algorithm}[p!]\label{alg:hist}
			\underline{\textbf{Computation of Figure (1)}}\\
			{\bf Input:} $n,k,u \ \ \ \ \  \ \ \ \ \ \ \ \ \ \ \ \ \ \ \ \ \ \ \ \ \ \ \ \ \ \ \ \ \ \ \ \ \ \ \ \ \ \  \ \ \ \ \ \ \ \ \ \ \ \ \ \ \ \ \ \ \ \ \ \ \ \ \ \ \ \ \ \ \ \ \ \ \ \ \ \ \ \ \ \ \ \ \ \ \ \ \ \ $ \\
			{\bf Output:} Histogram and normal QQ-plot of a random realization of $\phi(r_n^1),\ldots,\phi(r_n^k)$ 
			\begin{enumerate}
				\item \textbf{For} $j=1,\ldots,n$ \textbf{do}\\
				// Draw $(X_{1j},\ldots,X_{k_j})\sim N_k(0,I)$. \\
				// Set $Y_j=\sum_{i=1}^uX_{ij}$. \\
				\item \textbf{End for}.
				\item \textbf{For} $i=1,\ldots,k$ \textbf{do} \\
				// Compute the empirical correlation between the vectors $(X_{i1},\ldots,X_{in})$ and $(Y_1,\ldots,Y_n)$ and denote it by $r_n^i$.
				\item \textbf{End for}.
				\item \textbf{Return} Histogram and normal QQ-plot of $\phi(r_n^1),\ldots,\phi(r_n^k)$.		
			\end{enumerate} 
		\end{algorithm}
		
		\begin{figure}[p!] \label{fig:graph}
			\begin{center}
				
				\includegraphics[height=3.0in, width=5in]{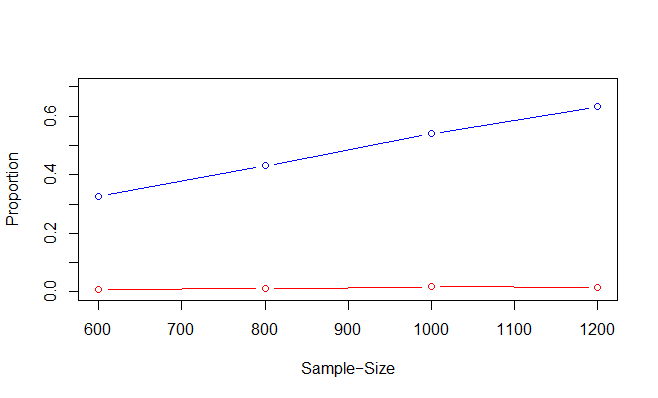} \caption{: This figure presents the estimated relations between the sample-size (horizontal axis) and the expected proportion of correct selections (vertical axis). The blue line is associated with the straightforward approach while the red one describes the estimates of the fast approximated approach. The parameters which were used in order to create this sketch are $B = 10, k = 20000,u = 100$ and the sample sizes are $n = 600, 800, 1000, 1200$. Notice that all of the standard deviations corresponding to this graph are smaller than $0.042$. 
				}
			\end{center}
		\end{figure}
		
		\begin{algorithm}[p!] \label{alg:graph}
			\underline{\textbf{Computation of Figure (2)}}\\
			{\bf Input:} $n, k, u, B \ \ \ \ \  \ \ \ \ \ \ \ \ \ \ \ \ \ \ \ \ \ \ \ \ \ \ \ \ \ \ \ \ \ \ \ \ \ \ \ \ \ \  \ \ \ \ \ \ \ \ \ \ \ \ \ \ \ \ \ \ \ \ \ \ \ \ \ \ \ \ \ \ \ \ \ \ \ \ \ \ \ \ \ \ \ \ \ \ \  $ \\
			{\bf Output:} Estimates of the expectation and standard deviation of the proportion of genes that are selected correctly as computed by the straightforward and fast approximated approaches.
			
			\begin{enumerate}
				\item \textbf{For} $t=1,\ldots,B$ \textbf{do}\\
				\ \ \ \ \ \ \textbf{For} $j=1,\ldots,n$ \textbf{do} \\
				\ \ \ \ \ \ // Draw $(X_{1j},\ldots,X_{k_j})\sim N_k(0,I)$ and set $Y_j=\sum_{i=1}^uX_{ij}$. \\
				\ \ \ \ \ \ \textbf{End for}. \\
				\ \ \ \ \ \ \textbf{For} $i=1,\ldots,k$ \textbf{do} \\
				\ \ \ \ \ \ // Compute the empirical correlation between the vectors $(X_{i1},\ldots,X_{in})$ and      $(Y_1,\ldots,Y_n)$ and denote it by $r_n^i(t)$. \\
				\ \ \ \ \ \ // Compute $d_t=\sum_{i=1}^uI_i$ where $I_i$ indicates whether $|r_n^i(t)|$ is one of the $u$ highest values of the vector $(|r_n^1(t)|,\ldots,|r_n^k(t)|)$. \\
				\ \ \ \ \ \ \textbf{End for}. \\
				\ \ \ \ \ \ // Compute the empirical variance of $r_n^1(t),\ldots,r_n^k(t)$ and denote it by $W$. 
				\\
				\ \ \ \ \ \  // Set $\hat{\sigma}_q=\sqrt{W-\frac{1}{n-3}}$ . \\ \ \ \ \ \ \  // Draw $\phi(\rho_1),\ldots,\phi(\rho_k)\stackrel{i.i.d}{\sim}N(0,\hat{\sigma}_q^2)$ . \\ 
				\ \ \ \ \ \  // Compute the set of indices which are associated with
				the $u$ highest values of the vector $(|\phi(\rho_1)|,\ldots,|\phi(\rho_1)|)$, denote it by $S_1$ and Draw $z_1(t),\ldots,z_k(t)\stackrel{i.i.d}{\sim}N(0,\frac{1}{n-3})$ . . \\
				\textbf{For} $i=1,\ldots,k$ \textbf{do} \\
				\ \ \ \ \ \ // Compute $v_i(t)=|r_n^i(t)+z_i(t)|$ . \\
				\ \ \ \ \ \ \textbf{End for}. \\
				\ \ \ \ \ \ // Compute the set of indexes which are associated with
				the $u$ highest values of the vector $(v_1(t),\ldots,v_k(t))$, denote it by $S_2$ and compute $c_t=\frac{|S_1\cap S_2|}{u}$. \\
				\textbf{End for}.\\ 
				\item Compute $\bar{d}=\frac{1}{B}\sum_{t=1}^Bd_t$ and $\bar{c}=\frac{1}{B}\sum_{t=1}^Bc_t$
				
				\item \textbf{Return} $\bar{d}$, $\bar{c}$, $\sqrt{\frac{1}{B}\sum_{t=1}^B(c_t-\bar{c})^2}$ and $\sqrt{\frac{1}{B}\sum_{t=1}^B(d_t-\bar{d})^2}$
			\end{enumerate}
		\end{algorithm}
		
\end{document}